\documentclass{article}
\usepackage[utf8]{inputenc}
\usepackage{amsthm,amsmath,amsfonts}
\usepackage{soul,color}
\usepackage{hyperref}
\hypersetup{
    colorlinks=true,
    linkcolor=magenta,
    }
\usepackage{xcolor,float}
\usepackage{tabularx}
\usepackage{longtable}
\usepackage{booktabs}
\usepackage{mathtools}
\usepackage{mathabx}
\usepackage{soul,color,xcolor}
\usepackage{tabularx}
\usepackage{longtable}
\usepackage{booktabs}
\usepackage{mathtools}
\usepackage{mathabx}
\usepackage{comment}
\DeclareMathOperator{\ord}{Ord}
\DeclareMathOperator{\sord}{Sord}
\newcommand{\Fq}{\mathbb{F}_q}

\newtheorem{theorem}{Theorem}[section]

\newtheorem{lemma}[theorem]{Lemma}
\newtheorem{conjecture}[theorem]{Conjecture}

\usepackage{color,soul,xcolor}

\usepackage{comment}
\theoremstyle{definition}
\newtheorem{definition}{Definition}[section]

\title{A Database of Quantum Codes}
\author{Nuh Aydin\footnote{Department of Mathematics, 
Kenyon College, aydinn@kenyon.edu}, Peihan Liu\footnote{Department of Mathematics, 
University of Michigan, paulliu@umich.edu}, Bryan Yoshino\footnote{Department of Mathematics, 
Kenyon College, yoshino1@kenyon.edu}}
\date{August, 2021}

\begin{document}

\maketitle

\begin{abstract}
Quantum error correcting codes (QECC) is becoming an increasingly important  branch of coding theory. For classical block codes, a comprehensive \href{codetables.de} 
{database} of best known codes exists which is available online at \cite{codetables}. The same database contains data on best known quantum codes as well, but only for  the binary field. There has been an increased interest in quantum codes over larger fields with many papers reporting such codes in the literature. However, to the best of our knowledge, there is no database of best known quantum codes for most fields.  We established a new database of QECC that includes codes over $\mathbb{F}_{q^2}$  for  $q\leq 29$. We also present several methods of constructing quantum codes from classical codes based on the CSS construction. We have found dozens of new quantum codes that improve the previously known parameters and also hundreds new quantum codes that did not exist in the literature.\\

\end{abstract}

\section{Introduction and Motivation}


Compared to classical information theory, the field of quantum information theory is relatively young. The idea of quantum error correcting codes was first introduced in \cite{Quantumoriginal1} and \cite{Quantumoriginal2}.  A method of constructing quantum error correcting codes (QECC) was given in \cite{Quantumoriginal3}. Since then researchers have investigated various methods of using classical error correcting codes to construct new QECCs. There are many publications in the literature with new quantum codes. However, compared to the classical block codes, the database of quantum codes is very limited. While \href{http://codetables.de/}{the database} (\cite{codetables}) covers the finite fields of orders $2,3,4,5,7,8$ and $9$, it only presents a database of QECC for $q=2$. There are \href{https://www.mathi.uni-heidelberg.de/~yves/Matritzen/QTBCH/QTBCHIndex.html} {static tables of some quantum codes} given at  \cite{quantumtwistedtables}, based on the work \cite{quantumtwisted}.  In recent years, researchers have conducted searches for good quantum codes, and some results are presented in the papers listed in the bibliography. We notice that the online tables \cite{quantumtwistedtables} may have been overlooked by some researchers. In many cases QECCs  over finite fields of characteristic greater than 2 are presented in the literature. Hence there is a need to have an interactive database of QECCs that contains best known quantum codes over larger alphabets which is continually updated. 
We now have compiled such a database. It is available at \url{http://quantumcodes.info}, and it is being updated continually. In addition to the creation of the database, we have also devised and implemented some
search algorithms to find new quantum codes from classical codes using the CSS construction. The paper also describes and reports the
results of these searches. One of the search methods we consider in this paper is about deriving quantum codes from polycyclic codes associated with trinomials, which is based on the CSS construction. This methods has been used in a recent work \cite{aydin2021polycyclic}. Here, we expand this method to construct more quantum codes.

This paper is organized as follows.  In section 2, we recall some basic concepts about quantum codes and  polycyclic codes. In section 3, we present some search algorithms along with the new quantum codes that we have found. A new code means it either has a better minimum distance than existing best known quantum code with the same length and dimension over the same alphabet  or a code with these parameters does not exist in literature. All computations are performed using the \href{http://magma.maths.usyd.edu.au/magma/} {MAGMA software}.

\section{Preliminaries}
Let $\mathbb{C}$ be the field of complex numbers and consider the Hilbert space $(\mathbb{C}^q)^{\otimes n}=\underbrace{\mathbb{C}^q\otimes \mathbb{C}^q\otimes\cdots \otimes\mathbb{C}^q}_{n \text{ times }}$. A $q$-ary quantum code of length and dimension $k$ is a $q^k$-dimensional subspace of $(\mathbb{C}^q)^{\otimes n}$. If the minimum distance of a quantum code is $d$ then we denote it by $[[n,k,d]]_q$. A connection between classical block codes and quantum codes was given in \cite{Quantumoriginal3} which has been used extensively in the literature. We include the relevant results here.

\begin{theorem}\cite{Quantumoriginal3}
Suppose $\bar{S}$ is an ($n - k$)-dimensional linear subspace of $\bar{E}$ (a $2n$-dimensional binary vector space) which is contained in its dual $\bar{S}^\perp$ (with respect to the inner product, and is such that there are no vectors of weight $\leq d - 1$ in $\bar{S}^\perp \setminus \bar{S}$. Then there is a quantum-error-correcting code mapping $k$ qubits to $n$ qubits which can correct $[(d - 1)/2]$ errors.
\end{theorem}
\begin{theorem}\cite{Quantumoriginal3}
Given two codes $[[n_1, k_1, d_1]]$ and $[[n_2, k_2, d_2]]$ with $k_2\leq n_1$ we can construct an $[[n_1 + n_2 - k_2, k_1, d]]$ code, where $d \geq \min\{d_1,d_1+d_2-k_2\}$.
\end{theorem}

The method of constructing quantum codes this way is known as CSS construction. There is a large number of papers in the literature that present new quantum codes obtained by the CSS construction. Recently, we have used polycyclic codes for the same purpose \cite{aydin2021polycyclic}.  

\begin{definition}
  A linear code $C$ is said to be polycyclic with respect to $v = (v_0, v_1, ... , v_{n-1}) \in \Fq^n$ if for any codeword $(c_0, c_1, ..., c_{n-1}) \in C$, its right polycyclic shift, $(0,c_0,c_1,\dots, c_{n-2})+c_{n-1}(v_0,v_1,\dots,v_{n-1})$ is also a codeword. Similarly, $C$ is  \textit{left} polycyclic with respect to $v = (v_0, v_1, ... , v_{n-1}) \in \Fq^n$ if for any codeword  $(c_0, c_1, ..., c_{n-1}) \in C$, its left polycyclic shift $(c_1,c_2,\dots, c_{n-1},0)+c_0(v_0,v_1,\dots,v_{n-1})$ is also a codeword.  If $C$ is both \textit{left} and \textit{right} polycyclic, then it is \textit{bi-polycyclic}.
\end{definition}
In this work, we only work with right polycyclic codes, which we simply refer to as polycyclic codes. Under the usual identification of vectors with polynomials, each polycyclic code $C$ of length $n$ is associated with a vector $v$ of length $n$ (or a polynomial $v(x)$ of degree less than $n$). We call $v$ ($v(x)$) an associate vector (polynomial) of $C$, and we say that $C$ is a polycyclic code associated with $x^n-v(x)$. Moreover, polycyclic codes of length $n$  associated with $f(x)=x^n-v(x)$ are ideals of the factor ring $\Fq[x]/\langle f(x) \rangle$. Note that an associate polynomial of a polycyclic code may not be unique.

Polycyclic codes are a generalization of cyclic codes and its several important generalizations. The following are some of the most important special cases of polycyclic codes:
\begin{itemize}
\item A right polycyclic code  with respect to $v=(1,0,0,...,0)$ is a cyclic code.\\
A left polycyclic code  with respect to $v=(0,0,0...,1)$ is a cyclic code.
\item A right polycyclic code with respect to $v=(-1,0,0,...,0)$ is a  negacyclic code. A left polycyclic code with respect to  $v=(0,0,0...,-1)$ is a negacyclic code.
\item A right polycyclic code with respect to $v=(a,0,0,...,0)$ is a constacyclic code. A left polycyclic code with respect to  $v=(0,0,0,...,a^{-1})$ is a constacyclic code.
\end{itemize}

\section{Quantum CSS Codes from Polycyclic Codes}

 Many of the works in the literature with new QECCs use  the CSS construction method given in \cite{Quantumoriginal3} or something related  to it. In this method, self-dual, self-orthogonal and dual-containing linear codes are used to construct quantum codes. The CSS construction requires two linear codes $C_1$ and $C_2$ such that $C_2^{\perp}\subseteq C_1$. Hence, if $C_1$ is a self-dual code, then we can construct a CSS quantum code using $C_1$ alone since $C_1^{\perp}\subseteq C_1$. If  $C_1$ is self-orthogonal, then we can construct a CSS quantum code with $C_1^{\perp}$ and $C_1$ since $C_1 \subseteq C_1^{\perp}$. Similarly in the case when $C_1$ is a dual-containing code.

\subsection{CSS Codes by Two Polycyclic Codes Associated with Trinomials}
By the definition of CSS construction, we need two codes such that one is contained in the dual of the other one. Let $C$ be a polycylic code generated by $g(x)$. By ideal inclusion we know that
\begin{align*}
    \langle g(x)f(x) \rangle \subseteq  \langle g(x) \rangle 
\end{align*}
for any polynomial $f(x)$and we consider $C_2^{\perp}=\langle g(x)f(x)\rangle \subseteq \langle g(x) \rangle =C_1 $. Here  $g(x)$ is a divisor of a trinomial $t(x)=x^n-ax^i-b$ and it generates a polycyclic code associated with $t(x)$. Then to find subcodes of $C_1$, we use the factorization of $t(x)$ into irreducibles. For each divisor $f(x)$ of $t(x)$ we get a subcode of $C_1$.  The following results from \cite{aydin2021polycyclic} are useful in formulating the search.

\begin{lemma} \cite{aydin2021polycyclic}
Fix $n,i\in \mathbb{Z}$, and let $a_1, a_2, b \in \Fq^{*}$ be nonzero with $a_1\not=a_2$. Then $\gcd(x^n-a_1x^i-b,x^n-a_2x^i-b)=1$
\end{lemma}

\begin{lemma}\cite{aydin2021polycyclic}
Fix $n,i\in \mathbb{Z}$, and let $a,b_1, b_2\in \Fq$, with $b_1\not=b_2$. Then $\gcd(x^n-ax^i-b_1,x^n-ax^i-b_2)=1$.
\end{lemma}

\begin{theorem} \cite{aydin2021polycyclic}
Let $n$ and $i$ be positive integers with $i<n$. For two distinct trinomials $x^n-ax^i-b$ and $x^n-cx^i-d$, we have $\gcd(x^n-ax^i-b,x^n-cx^i-d)$ is either 1 or a binomial of degree $\gcd(n,i)$.
\end{theorem}

\begin{lemma} \cite{aydin2021polycyclic}
For any $x^n-ax^i-b\not=x^n-cx^i-d\in \mathbb{F}_3[x]$, we have  $\gcd(x^n-ax^i-b,x^n-cx^i-d)=1$.
\end{lemma}


We present some examples of quantum codes with new parameters in Table 1 where $t(x)$ refers to the trinomial, $g_1(x)$ is the generator of $C_1$, and $h_2(x)$ is the generator of $C_2^\perp$, i.e. $\langle g_1(x) \rangle =C_1$ and $\langle g_1(x)f(x) \rangle =\langle h_2(x) \rangle =C_2^\perp$. The last column in the table refers to the literature where a comparable code is presented. Every code in this table has a higher minimum distance than the comparable code presented in the literature. In the tables below, the coefficients of a polynomial are listed in increasing powers of $x$, with leading coefficient at the right. Also, for alphabet of size greater than 10, $A$ denotes $10$, $B$ denotes $11$, ..., $H$ denotes $17$. Hence the entry $[79F1]$ for $g_1$ on the first row of Table 2 represents the polynomial $7+9x+15x^2+x^3$

\begin{table}[H]
\resizebox{.8\textwidth}{!}{\begin{minipage}{\textwidth}

\caption{\small{New quantum codes  constructed from polycyclic codes associated with trinomials}}
\begin{tabular}{p{2.2cm}p{2.2cm}p{2.8cm}p{5.2cm}p{.7cm}}
\\\hline\noalign{\smallskip}
$[[n,k,d]]_{q^2}$  &$t$& $g_1$ & $h_2$ & Ref.\\ 
\noalign{\smallskip}\hline\noalign{\smallskip}
  $[[60, 36, 5]]_{3^2}$ & $x^{60}-2x^{6}-2$  & $[1121011011001]$ & $[1222021012001]$ & \cite{22}\\
  [.5ex]
  $[[81, 52, 5]]_{3^2}$ & $x^{81}-x^{30}-1$  & $[200210101101001]$ & $[1021111212022001]$ & \cite{16}\\
  [.5ex]
  $[[80, 54, 4]]_{5^2}$ & $x^{80}-x-1$  & $[4231342331]$ & $[101124230104302301]$ & \cite{8932583}\\
  [.5ex]
  $[[96, 80, 4]]_{5^2}$ & $x^{96}-x-1$  & $[310032201]$ & $[101110301]$ & \cite{17}\\
  [.5ex]
  $[[52, 28, 4]]_{11^2}$ & $x^{52}-2x^{2}-2$  & $[7263851]$ & $[84782120107151973361]$ & \cite{12}\\
  [.5ex]
  $[[11, 1, 5]]_{13^2}$ & $x^{11}-x-3$  & $[8A0C31]$ & $[352721]$ & \cite{quantumFromConstacyclic}\\
  [.5ex]
  $[[7, 1, 4]]_{17^2}$ & $x^{7}-x^{2}-2$  & $[BFC1]$ & $[99101]$ & \cite{quantumFromConstacyclic}\\
  [.5ex]
\noalign{\smallskip}\hline
\end{tabular}
\end{minipage}}
\end{table}


\subsection{CSS Codes by a Polycyclic Code Associated with Trinomials and a Random Linear Code}
In the search above, $C_1$ and $C_2$ are both polycyclic codes associated with the same trinomial. However, this is not required. Hence, we can find some $f(x)$ with desired degrees, and then construct $C_2$ by $g(x)f(x)$, where $C_1=\langle g(x) \rangle$ and $C_2=\langle g(x)f(x) \rangle$, i.e., $C_2$ is not neccessarily polycyclic. 

Note that this method works particularly well for \textit{target search} by which we mean, we fix the field $q$, length $n$ and dimension $k$, and find the best known minimum distance for this set of parameters. Next, we randomly generate a number of polynomials $f(x)$ of degree $k$ and iterate through divisors $g(x)$ of $t(x)$, where $t(x)$ is a trinomial, so that the CSS code constructed by $C_1=\langle g(x) \rangle$ and $C_2^\perp =\langle g(x)f(x) \rangle$ is in the form $[[n,k]]_{q^2}$. 

\begin{table}[H]
\resizebox{.8\textwidth}{!}{\begin{minipage}{\textwidth}

\caption{\small{New quantum codes  constructed from a polycyclic code and a random linear code}}
\begin{tabular}{p{2.2cm}p{2.1cm}p{2.1cm}p{8.8cm}p{.7cm}}
\\\hline\noalign{\smallskip}
$[[n,k,d]]_{q^2}$  &$t$& $g_1$ & $g_2$ & Ref.\\ 
\noalign{\smallskip}\hline\noalign{\smallskip}
  $[[7, 1, 4]]_{17^2}$ & $x^{17}-x-1$  & $[79F1]$ & $[36DD1]$ & \cite{quantumFromConstacyclic} \\
  [.5ex]
  $[[24, 18, 3]]_{17^2}$ & $x^{24}-x-1$  & $[EG41]$ & $[A6872366CA28406AE407F1]$ & \cite{Fq+v1Fq+...+vrFq}\\
  [.5ex]
  $[[36, 30, 3]]_{17^2}$ & $x^{36}-x-1$  & $[3791]$ & $[94F9161A374963B0E663E35AA4A6EFDGG1]$ & \cite{Fq+v1Fq+...+vrFq} \\
  [.5ex]
  $[[48, 36, 4]]_{17^2}$ &  $x^{48}-x-3$ & $[CD3BE1]$& $[60942809G709FE411F4910DCC6A7B5BC1545D270681]$ & \cite{Fq+v1Fq+...+vrFq}\\
  [.5ex]
\noalign{\smallskip}\hline
\end{tabular}
\end{minipage}}
\end{table}


\subsection{CSS Construction by Two Polycyclic Codes Associated with Multinomials}

It is also possible to expand the search space by considering multinomials instead of trinomials. Even though data suggests that binomials and trinomials have more divisors than other multinomials, we still found some good codes using this  search. Some good quantum code have been found using this method in \cite{aydin2021polycyclic}.


\begin{table}[H]
\resizebox{.8\textwidth}{!}{\begin{minipage}{\textwidth}

\caption{\small{New quantum codes constructed from multinomials}}
\begin{tabular}{p{1.5cm}p{8.8cm}p{3.2cm}p{3cm}p{.7cm}}
\\\hline\noalign{\smallskip}
$[[n,k,d]]_{q^2}$  &$m$& $g_1$ & $g_2$ & Ref.\\ 
\noalign{\smallskip}\hline\noalign{\smallskip}
  $[[36, 6, 5]]_{3^2}$ & $[2011110210112212022010120211000010112102000101111]$  & $[222101212100200021]$ & $[222101212100200021]$ & \cite{skew}\\
  [.5ex]
  $[[22, 2, 6]]_{5^2}$ & $[13343122443414240410122]$  & $[4223310221]$ & $[442143133401]$ & \cite{quantumtwistedtables}\\
  [.5ex]
\noalign{\smallskip}\hline
\end{tabular}
\end{minipage}}
\end{table}

\section{Obtaining New Codes from Existing Codes}
Like classical codes, we can also obtain new quantum codes from existing ones using the standard method of extending (E), shortening (S) or puncturing (P) a code or taking the direct sum (DS) of two given codes. All of these standard constructions are available in the Magma software.  We applied these constructions on exiting coded and  obtained many new codes that have better minimum distances  than the codes with the same length and dimension presented in the literature. In many cases we started with the codes in the tables \cite{quantumtwistedtables}.


\begin{longtable}{p{2.2cm}p{5cm}p{2.4cm}p{0.7cm}}
\caption{\small{New Quantum Codes from Existing Codes}}
\\\hline\noalign{\smallskip}
$[[n,k,d]]_{q^2}$& Existing Code(s) (Method) &$[[n',k',d']]_{q^2}$ &Ref.\\ 
\noalign{\smallskip}\hline\noalign{\smallskip}
$[[41, 1, 10]]_{3^2}$ & $[[33, 1, 10]]$ (E) &$[[41, 1, 8]]_{3^2}$ &\cite{quantumtwistedtables} \\
  [.5ex]
  $[[45, 21, 7]]_{3^2}$ & $[[41, 21, 7]]$ (E) & $[[45, 21, 5]]_{3^2}$ &\cite{22} \\
  [.5ex]
  $[[65, 29, 6]]_{3^2}$ & $[[25, 7, 6]]$, $[[40, 22, 6]]$ (DS)& $[[65, 29, 5]]_{3^2}$ &\cite{quantumtwistedtables} \\
  [.5ex]
  $[[119, 7, 13]]_{3^2}$ & $[[52, 3, 13]]$, $[[67, 4, 13]]$ (DS)& $[[119, 7, 11]]_{3^2}$ &\cite{quantumtwistedtables} \\
  [.5ex]
  $[[120, 6, 14]]_{3^2}$ & $[[53, 3, 14]]$, $[[67, 3, 14]]$ (DS)& $[[20, 6, 12]]_{3^2}$ &\cite{quantumtwistedtables} \\
  [.5ex]
  $[[120, 96, 6]]_{3^2}$ & $[[121, 96, 7]]$ (P)& $[[120, 96, 4]]_{3^2}$ &\cite{22} \\
  [.5ex]
  $[[34, 2, 8]]_{5^2}$ & $[[26, 2, 8]]$ (E)& $[[34, 2, 6]]_{5^2}$ &\cite{quantumtwistedtables} \\
  [.5ex]
  $[[38, 2, 8]]_{5^2}$ & $[[26, 2, 8]]$ (E)& $[[38, 2, 7]]_{5^2}$ &\cite{quantumtwistedtables} \\
  [.5ex]
  $[[80, 48, 5]]_{5^2}$ & $[[11, 1, 5]]$, $[[69, 47, 6]]$ (DS)& $[[80, 48, 2]]_{5^2}$ &\cite{8} \\
  [.5ex]
  $[[80, 54, 7]]_{5^2}$ & $[[78, 54, 7]]$ (E)& $[[80, 54, 3]]_{5^2}$ &\cite{8932583} \\
  [.5ex]
  $[[80, 56, 5]]_{5^2}$ & $[[16, 4, 5]]$, $[[64, 52, 5]]$ (DS)& $[[80, 56, 3]]_{5^2}$ &\cite{19} \\
  [.5ex]
  $[[88, 8, 12]]_{5^2}$ & $[[52, 8, 12]]$ (E)& $[[88, 8, 5]]_{5^2}$ &\cite{8} \\
  [.5ex]
  $[[88, 48, 7]]_{5^2}$ & $[[22, 2, 7]]$, $[[66, 46, 7]]$ (DS)& $[[88, 48, 2]]_{5^2}$ &\cite{6} \\
  [.5ex]
  $[[99, 3, 11]]_{5^2}$ & $[[29, 1, 11]]$, $[[70, 2, 11]]$ (DS)& $[[99, 3, 9]]_{5^2}$ &\cite{quantumtwistedtables} \\
  [.5ex]
  $[[27, 17, 5]]_{7^2}$ & $[[25, 17, 5]]$ (E)& $[[27, 17, 3]]_{7^2}$ &\cite{6} \\
  [.5ex]
  $[[49, 40, 5]]_{7^2}$ & $[[48, 40, 5]]$ (E)& $[[49, 40, 3]]_{7^2}$ &\cite{16} \\
  [.5ex]
  $[[60, 36, 5]]_{7^2}$ & $[[18, 8, 5]]$, $[[42, 28, 5]]$ (DS)& $[[60, 36, 4]]_{7^2}$ &\cite{6} \\
  [.5ex]
   $[[55, 35, 6]]_{7^2}$ & $[[58, 35, 9]]$ (P)& $[[55, 35, 4]]_{7^2}$ &\cite{quantumtwistedtables} \\
  [.5ex]
 \noalign{\smallskip}\hline
\end{longtable}

\section{Features of the Database}

We introduce an interactive, \href{http://quantumcodes.info/} {online database} \cite{quantumcodesdatabase}  which stores the parameters of best known $[[n, k, d]]_{q^2}$ QECCs over fields $GF(q^2)$ for $q \leq 29$ where $q$ is a power of a prime and $n \leq 200$. Users of the database can search for best known QECCs by fixing $q$, $n$ and $k$ or by fixing $q$, $n$ and $d$. Also, users can search for multiple QECCs at the same time by fixing $q$ and searching over ranges for $n$ and $d$. Every QECC displayed in the search results has a reference which is a link to the paper in which the code was introduced or an online table where it is listed. In addition, all MDS codes are marked with an asterisk.

To construct such a database, we surveyed recent literature on QECCs and kept a record of parameters of codes in the literature. The bibliography gives the complete list of papers that have been examined. %

\end{document}